\documentclass[a4paper,11pt]{article}
\usepackage{algorithm}
\usepackage[noend]{algorithmic}
\usepackage{latexsym}

\newtheorem{theorem}{Theorem}
\newtheorem{lemma}{Lemma}
\newenvironment{proof}{\noindent{\it Proof. }\ignorespaces}
{\hspace*{\fill}$\Box$\par\medskip}

\title{Is this the simplest (and most surprising) sorting algorithm ever?}

\author{Stanley P. Y. Fung %
\footnote{School of Computing and Mathematical Sciences, 
University of Leicester, United Kingdom.
Email: pyf1@le.ac.uk} }

\date{}

\begin{document}
\maketitle

\begin{abstract}
We present an extremely simple sorting algorithm. It may look like it is
obviously wrong, but we prove that it is in fact correct. 
We compare it with other simple sorting algorithms, and analyse some of 
its curious properties.
\end{abstract}

\section{The Algorithm}

Most of us know those simple sorting algorithms like bubble sort very well.
Or so we thought -- 
have you ever found yourself needing to write down the pseudocode of
bubble sort, only to realise that it is not as straightforward as you think
and you couldn't get it right the first time? 
It needs a bit of care to get the loop indices start and 
end at just the right values, not going out of bounds, or to handle some 
flag variables correctly. 
Wouldn't it be nice if there is a simple algorithm with no such hassle?

Here is such an algorithm, that sorts an array $A$ of $n$ elements in 
non-decreasing order. For easier exposition in the proof later on, 
the array is 1-based, so the elements are $A[1], \ldots, A[n]$.

\begin{algorithm}[H]
\caption{ICan'tBelieveItCanSort($A[1..n]$)}
\begin{algorithmic}
\FOR{$i = 1$ \TO $n$}
  \FOR{$j = 1$ \TO $n$}
    \IF{$A[i] < A[j]$} \STATE swap $A[i]$ and $A[j]$ \ENDIF
  \ENDFOR
\ENDFOR
\end{algorithmic}
\label{alg:1}
\end{algorithm}

That's it, really. Just loop over every pair of $(i,j)$ values in the 
standard double-for-loop way, and compare and swap. 
What can possibly be simpler?

The first reaction of someone seeing this algorithm 
might be ``this cannot possibly be correct'', or 
``you got the direction of the inequality wrong and also the loop indices
wrong''. But no, it does sort correctly in increasing order.
The algorithm may give the impression that it sorts in decreasing order 
because it swaps when $A[i]$ is {\it smaller} than $A[j]$. However, unlike
some other similar-looking sorting algorithms, the loops do not constrain 
$j$ to be larger than $i$. In fact, it can be seen from the proof later that 
most of the swaps only happen when $i>j$, so $A[i]<A[j]$ is actually
an inversion that needs swapping.

There is nothing good about this algorithm.
It is slow -- the algorithm obviously runs in $\Theta(n^2)$ time, 
whether worst-case, average-case or best-case. 
It unnecessarily compares all pairs of positions, twice (but see
Section~3).
There seems to be no intuition behind it, and
its correctness is not entirely obvious.
You certainly do not want to use it as a first example to introduce 
students to sorting algorithms.
It is not stable, 
does not work well for external sorting, 
cannot sort inputs arriving online, 
and does not benefit from partially sorted inputs.
Its only appeal may be its simplicity, in terms of lines of code and 
the ``symmetry'' of the two loops.

It is difficult to imagine that this algorithm was not discovered before,
but we are unable to find any references to it.

\section{The Proof}

Although at first counterintuitive, after a few moments of thought 
it should not be difficult to see why the algorithm is correct. 
Nevertheless, we give an overly elaborate proof here.

\begin{theorem}
Algorithm~\ref{alg:1} sorts the array in non-decreasing order.
\end{theorem}
\begin{proof}
We show by induction the following property, for $1 \le i \le n$:

\begin{quote}
(${\bf P}_i$): Right after the $i$-th iteration of the outer loop, the first $i$
elements in $A$ are in non-decreasing order, i.e.
$A[1] \le A[2] \le \ldots \le A[i]$;
furthermore $A[i]$ is the maximum of the whole array.
\end{quote}

The correctness of the algorithm clearly will follow from (${\bf P}_n$).

When $i=1$, the algorithm sequentially checks each element in $A[2..n]$,
swapping it into $A[1]$ if it is bigger than $A[1]$. It is clear that by the 
end of this process, $A[1]$ contains the largest element of $A$.
Hence (${\bf P}_1$) is true.

Suppose (${\bf P}_i$) is true, and so right after the $i$-th iteration of the
outer loop, we have $A[1] \le \ldots \le A[i]$. 
Now we consider the $(i+1)$-th iteration and prove (${\bf P}_{i+1}$).
Let $k \in [1,i]$ be the smallest index such that 
$A[i+1] < A[k]$ (and thus $A[k-1] \le A[i+1]$ as long as $k \neq 1$).
If $A[i+1] \ge A[i]$, no such $k$ satisfies the criteria and we set $k=i+1$.
We consider three stages of the inner $j$ loop:

\begin{itemize}
\item When $1 \le j \le k-1$, the algorithm compares $A[i+1]$ with each
of $A[1], \ldots, A[k-1]$. Since $A[i+1]$ is at least as large as each of
them, no swap takes place.

\item Next consider $k \le j \le i$.
(If $k=i+1$, this stage does not exist.)
We argue that (1) every comparison will result in a swap (unless the two
elements are equal, in which case it doesn't matter); (2) after iteration $j$,
$A[1..j]$ is in sorted order; and (3) at each iteration $j \neq i$, 
the element swapped out into position $i+1$ is at least as large as all 
the elements in $A[1..j]$ (including the one just swapped into $A[j]$) 
but not larger than $A[j+1]$.

When $j=k$ (and $\neq i$),
since $A[i+1]<A[k]$ by definition of $k$, the algorithm swaps them. 
For clarity, we use $A[\ ]$ and $A'[\ ]$ to denote the array elements before
and after the swap, so $A'[i+1]=A[k]$ and $A'[k]=A[i+1]$.
We have $A'[k] = A[i+1] \ge A[k-1] = A'[k-1]$
and so $A'[1..k]$ is in sorted order.
Also $A'[i+1] = A[k] > A[i+1] = A'[k]$
and $A'[i+1] = A[k] \le A[k+1] = A'[k+1]$. Thus the three properties are
satisfied.

Similarly, when $j=k+1$ (and $\neq i$), 
$A'[i+1]$ is compared with $A'[k+1]$.
Since $A'[i+1] \le A'[k+1]$, the algorithm swaps them
(unless they are equal). 
We use $A''[\ ]$ to denote the array contents after this swap.
Whether they are swapped or not, we have
$A''[k+1] = A'[i+1] \ge A'[k] = A''[k]$ and so
$A''[1..k+1]$ is in sorted order.
Also $A''[i+1] = A'[k+1] \ge A'[i+1] = A''[k+1]$
and $A''[i+1] = A'[k+1] \le A'[k+2] = A''[k+2]$.
Thus the three properties are again satisfied.

The same happens for each subsequent step. 
Finally, when $j=i$, the same happens regarding property (1), (2) and the
first part of (3). The second part of (3) is not needed since it is only
used to establish the properties of the next iteration within this stage.
Note that by (${\bf P}_i$), 
$A[i]$ is the largest element in $A$, so after this step
$A[i+1]$ must be the largest element in $A$.

\item When $i+1 \le j \le n$: By this point,
$A[1] \le A[2] \le \ldots \le A[i+1]$ and
$A[i+1]$ is the largest element in $A$. Hence all further comparisons 
between $A[i+1]$ and $A[j]$ will not result in any swaps.

\end{itemize}

Thus we established that (${\bf P}_{i+1}$) is true.
\end{proof}

\section{Further comments}

\subsection{Relation to other sorting algorithms}

Sorting algorithms are often classified into different types such as
exchange-based, selection-based, insertion-based, etc (see \cite{TAOCP}).
For reference, this is the standard exchange sort:

\begin{algorithm}[H]
\caption{ExchangeSort($A[1..n]$)}
\begin{algorithmic}
\FOR{$i = 1$ \TO $n$}
  \FOR{$j = i+1$ \TO $n$}
    \IF{$A[i] > A[j]$} \STATE swap $A[i]$ and $A[j]$ \ENDIF
  \ENDFOR
\ENDFOR
\end{algorithmic}
\label{alg:exc}
\end{algorithm}

Algorithm~\ref{alg:1} was found when the author was making up some wrong
sorting algorithms and changed the $j$ loop of Algorithm~\ref{alg:exc} to
be from 1 to $n$, and was shocked to see that it sorts in decreasing order.

While the pseudocode of Algorithm~\ref{alg:1} makes it look like an 
exchange-based sort,
in fact the first iteration of the outer loop ($i=1$) selects the maximum,
and then each of the other iterations works like insertion sort.
Thus in some way it is a combination of selection-based and
insertion-based algorithms.

\subsection{``Improving'' the algorithm}

It can be seen from the proof that, other than the $i=1$ iteration,
the inner $j$ loop needs to be executed from 1 to $i-1$ only; the rest 
of it has no effect.
But the only thing the $i=1$ iteration does is to extract the maximum and
move it to $A[1]$, and clearly there is no reason why a sorting algorithm 
(for sorting in increasing order) would want to do that.
The extraction of the maximum gives the the second part of property
(${\bf P}_i$), but that part is never needed in the proof except to show that
the comparisons of $A[i]$ with any element in $A[i+1..n]$ will not result in 
swaps. But if they do swap, it will only make $A[i]$ larger, and
$A[i+1..n]$ is not covered by (${\bf P}_i$) so 
whatever happen to those elements don't matter.

Thus we can remove those unnecessary iterations, resulting in the 
following ``improved'' algorithm which also sorts correctly and makes fewer 
comparisons and swaps:

\begin{algorithm}[H]
\caption{(``Improvement'' to Algorithm~\ref{alg:1})}
\begin{algorithmic}
\FOR{$i = 2$ \TO $n$}
  \FOR{$j = 1$ \TO $i-1$}
    \IF{$A[i] < A[j]$} \STATE swap $A[i]$ and $A[j]$ \ENDIF
  \ENDFOR
\ENDFOR
\end{algorithmic}
\label{alg:ins}
\end{algorithm}

And we have re-discovered insertion sort!

The standard insertion sort finds the insertion point from the end of the 
sorted region, shifting elements along the way, and stops as soon as the
correct insertion point is found. In contrast,
Algorithm~\ref{alg:ins} checks all elements in the sorted region, starting 
from the beginning, using repeated swaps instead of shifts to move the 
elements, and always goes through the full length of the sorted region.
So Algorithm~\ref{alg:ins} is slower than the standard insertion sort, and 
Algorithm~\ref{alg:1} is even slower.
Nevertheless we prefer our original Algorithm~\ref{alg:1}
for its ``beauty''.

\subsection{Sort in decreasing order}

To sort in non-increasing order instead, one can obviously reverse the
inequality sign in Algorithm~\ref{alg:1}. Alternatively, due to the 
symmetry of $i$ and $j$, swapping the two for loops will achieve the same 
effect.

\section{Best-case and worst-case inputs}

What are the best-case and worst-case inputs of Algorithm~\ref{alg:1}?
Obviously it always makes exactly $n^2$ comparisons, so we consider the
number of swaps. Most sorting algorithms make at most
$n(n-1)/2$ comparisons/swaps. After all, that is the number of pairs
of different elements you can compare, 
and also the maximum number of inversions.
(Recall that an inversion is a pair $(i,j)$ where $i<j$ but $A[i]>A[j]$.)
It would not be surprising that this algorithm would do worse.
Indeed, we prove that it can make more swaps than the maximum number 
of inversions -- but by just 1.

In the rest of this section, we assume all the elements are distinct.
Furthermore, since this is a comparison-based algorithm, 
we can without loss of generality assume the elements are (a permutation of)
$1, 2, \ldots, n$.
Recall from Sections 3.1 and 3.2 that the algorithm can be considered to have 
two phases: the selecting the maximum phase ($i=1$) and the insertion sort
phase ($i=2$ to $n$).

\begin{lemma}
Each swap in the selection phase increases the number of inversions by
exactly one, and each swap in the insertion phase decreases the number of
inversions by exactly one.
\end{lemma}
\begin{proof}
In the selection phase, suppose $A[1]$ is being swapped with $A[j]$.
This pair goes from non-inversion to inversion.
Due to how the selection phase works, all numbers
between $A[1]$ and $A[j]$ at this point are smaller than $A[1]$ 
and hence also smaller than $A[j]$.
Thus, for any pair of indices $(i_1, i_2)$ where $i_1 \in \{1,j\}$
and $i_2 \in (1,j)$, its inversion status does not change.
The number of inversions involving pairs where $i_2 > j$ does not change 
either.

Similarly, in the insertion phase, if $A[i]$ and $A[j]$ are being swapped,
all other numbers between $A[i]$ and $A[j]$ are larger than both of them 
due to how the algorithm works. Hence the only change in the number of
inversions comes from the pair itself, which changes from inversion to
non-inversion.
\end{proof}

\begin{theorem}
Algorithm~\ref{alg:1} makes at most $I_{max}+1$ swaps, where
$I_{max} = n(n-1)/2$ is the maximum number of inversions possible in
any input, and this bound is tight.
\end{theorem}
\begin{proof}
Consider where the maximum element $n$ is in the input.
\begin{itemize}
\item
If $n$ is in $A[1]$, then the selection phase makes no swaps. 
The insertion phase makes at most $I_{max}$ swaps, since each swap
removes one inversion.
\item
If $n$ is in $A[2]$, then the selection phase makes one swap, and
the insertion phase makes at most $I_{max}$ swaps. Thus the total
is at most $I_{max}+1$.
\item
Suppose $n$ is in $A[j]$, $j \ge 3$. At step $j$ of the selection phase,
it will be swapped to $A[1]$.
Consider where the second largest number $n-1$ is just before this swap. 
If it is in $A[1..j-1]$, then it must in fact be in $A[1]$. This is
because after step $j-1$, $A[1]$ must contain the largest number in 
$A[1..j-1]$ due to how the selection phase works.
At step $j$, this number $n-1$ is swapped to position
$j$ of $A$. But the maximum-inversion configuration is
$[n, n-1, \ldots, 2, 1]$ which puts $n-1$ in position 2.
For each position that $n-1$ is put further away from its ``ideal'' position 2,
the number of inversions is reduced from $I_{max}$ by at least 1. 
Thus at this point the number of inversions is at most
$I_{max}-(j-2)$.
No further swaps will happen in this phase, so this is also the 
number of inversions at the end of the selection phase.

Otherwise, if $n-1$ is in $A[j+1..n]$ just before step $j$, then
step $j$ will swap some other number into position $j$, but again no 
further swaps will take place, and so $n-1$ is even more positions
away from position 2. Thus the number of inversions can only be smaller.

In the selection phase the algorithm makes at most $j-1$ swaps,
and in the insertion phase it makes at most $I_{max}-(j-2)$ swaps.
Thus the total is at most $I_{max}+1$.
\end{itemize}

Two inputs attain this bound: $[n-1, n, n-2, n-3, \ldots, 2, 1]$, or
$[n-2, n-1, n, n-3, n-4, \ldots, 2, 1]$. In other words, it places the 
largest two/three numbers in increasing order, followed by the rest in
decreasing order.
These are the only two inputs attaining the tight bound; this can be
proved by a slightly tighter analysis for $j \ge 4$ along similar lines.
\end{proof}

If the input has few inversions, Algorithm~\ref{alg:1} is ``adaptive'' to 
it, in a way:

\begin{theorem}
Algorithm~\ref{alg:1} makes at most $I + 2(n-1)$ swaps, where $I$ is
the number of inversions in the input, and this bound is tight.
\end{theorem}
\begin{proof}
The algorithm can make at most $n-1$ swaps in the selection phase, raising
the number of inversions to at most $I + (n-1)$. Then the insertion phase
makes at most $I+(n-1)$ swaps, since each swap removes one inversion.
Thus the total number of swaps is at most $I+2(n-1)$.

The bound is attained when the input is already sorted, i.e.
$A = [1, 2, \ldots, n]$. The number of inversions increases from 0 to
$n-1$ in the selection phase.
\end{proof}

\begin{theorem}
Algorithm~\ref{alg:1} makes at least $n-1$ swaps, and this bound is tight.
\end{theorem}
\begin{proof}
Just after the selection phase, the maximum element is in $A[1]$.
Hence there are at least $n-1$ inversions at this point. 
Thus the insertion phase requires at least $n-1$ swaps.

The only input attaining this bound is
$A = [n, 1, 2, \ldots, n-1]$. The selection phase makes no swaps and
the insertion phase makes $n-1$ swaps.
\end{proof}

\end{document}